\documentclass[a4paper,UKenglish,cleveref, autoref, thm-restate]{lipics-v2021}

\pdfoutput=1
\hideLIPIcs

\title{Simple Nash Equilibria for Qualitative Multiplayer Games}

\author{Mona Alluwaym}{University of Liverpool, United Kingdom}{}{https://orcid.org/0009-0007-5661-7885}{}

\author{{James C.~A.} Main}{University of Oxford, United Kingdom}{}{https://orcid.org/0009-0000-8471-4833}{}

\author{Sven Schewe}{University of Liverpool, United Kingdom}{}{https://orcid.org/0000-0002-9093-9518}{}

\authorrunning{Mona Alluwaym, James C.\ A.\ Main, and Sven Schewe} 

\Copyright{Mona Alluwaym, James C.~A.~Main and Sven Schewe} 

\ccsdesc[100]{Theory of computation~Solution concepts in game theory} 

\keywords{games on graphs, multiplayer games, Nash equilibria, subgame-perfect equilibria, memoryless strategies} 

\category{} 

\relatedversion{}

\funding{This work has been supported by the EPSRC projects EP/Z536179/1, EP/X03688X/1, and EP/X042596/1.}

\nolinenumbers

\usepackage[T1]{fontenc}
\usepackage{graphicx}
\usepackage{amsmath,amssymb,mathtools}
\usepackage[disable]{todonotes}
\usepackage{tikz}
\usetikzlibrary{automata,positioning,arrows.meta,fit,backgrounds,shapes.geometric,calc}

\pgfdeclarelayer{background}
\pgfsetlayers{background,main}

\newcommand{\llb}{[\![}
\newcommand{\rrb}{]\!]}
\newcommand{\G}{\mathcal{G}}

\newcommand{\Player}{\mathcal{P}}
\newcommand{\Reach}{\mathrm{Reach}}
\newcommand{\Safe}{\mathrm{Safe}}

\begin{document}

\maketitle

\begin{abstract}
We investigate memory requirements for Nash and subgame-perfect equilibria in turn-based deterministic games with $\omega$-regular objectives.
We prove that memoryless randomised (i.e., stationary) subgame-perfect equilibria always exist in games with reachability, safety, and 0-2 Muller objectives (i.e., Muller objectives for which accepting sets are either up- or downward closed), and any combination of these objectives.
We provide an algorithm to construct such an equilibrium.
We also show that randomisation may be required to construct memoryless equilibria in games with reachability or B\"uchi as well as safety or CoB\"uchi objectives, and that memoryless equilibria need not exist for any other class of Muller objectives (with respect to the Mostowski hierarchy).
\end{abstract}

\todo[inline]{The title page does not count toward the page limit, hence the page break.}

\section{Introduction}
\subparagraph*{Games on graphs.}
We study turn-based games played on finite graphs (e.g.,~\cite{DBLP:conf/dagstuhl/2001automata,gog23}).
The structure on which the game is played is called a turn-based deterministic arena (or simply an arena); it is described by a directed graph and a partition of the vertices among the players.
At the beginning of a play, a token is placed on an initial vertex and, in each round, the player in control of the current vertex moves the token along an outgoing edge.
The players continue this interaction for infinitely many rounds and construct a play, i.e., an infinite path of the graph.
The goal of a player is formalised by an objective, i.e., a subset of plays.
Players follow strategies describing how to select moves.
Formally, strategies are functions assigning to any play history the decisions to be made.
In full generality, strategies may use memory and randomisation.
A strategy is called pure if it deterministically assigns moves to histories, memoryless (or stationary) if its decisions depend only on the current vertex and not the full history, and positional if it is both memoryless and pure.

Games on graphs are notably useful to solve the controller synthesis problem (e.g.,~\cite{DBLP:reference/mc/BloemCJ18}).
Given a reactive system, the goal of controller synthesis is to automatically construct a controller for the system that enforces some specification no matter the behaviour of the environment of the system.
It can be solved by modelling the interaction of the system and its environment as a zero-sum game on a graph: the goal of the system player is to enforce their objective (which formalises the specification) and the environment player aims to falsify this objective.
Strategies of the system player correspond to controllers; thus the goal is to find a winning strategy, i.e., a strategy enforcing the objective of the system against all strategies of its adversary.
In general, we are interested in \emph{simple winning strategies} (e.g., using limited memory), as they yield simpler controllers (see~\cite{DBLP:conf/rp/Randour25} and references therein).

\subparagraph*{Multiplayer games.}
For systems consisting of several components interacting together, treating each component as competing with the others may be too restrictive (e.g., it precludes cooperation between components).
Instead, we consider games in which each component is a player.
Each player has their own objective and these are not necessarily incompatible.
Such games are called multiplayer \emph{non-zero-sum games}.

\emph{Nash equilibria}~\cite{Nash50} (NE) are the most classical formalisation of rational behaviour in such games.
An NE is a strategy profile, i.e., a tuple of one strategy per player, such that no player can increase the probability that their objective holds by unilaterally changing their strategy.
\emph{Subgame-perfect equilibria}~\cite{Sel65} (SPE) are another well-studied notion of rationality.
Intuitively, SPEs are a refinement of NEs requiring that from all histories of the game (even those that are not reachable with the considered strategy profile), the players follow an NE.
This stronger requirement excludes NEs based on non-credible threats, i.e., where a player threatens to make a move that is detrimental to the satisfaction of their objective after some history to discourage other players from moving to some vertex.

We focus on the complexity of equilibria in games with $\omega$-regular objectives.
The class of $\omega$-regular languages is fundamental in the fields of synthesis and verification~\cite{BK08}, and includes, notably, linear temporal logic~\cite{DBLP:conf/focs/Pnueli77}.
In games with $\omega$-regular objectives, there exist pure SPEs such that all strategies are \emph{finite-memory}~\cite{DBLP:conf/fsttcs/Ummels06} -- see also~\cite{DBLP:conf/lfcs/BrihayePS13,depril2013} for a proof of existence of NE with little memory (for a class of objectives subsuming $\omega$-regular objectives).

We focus on three classes of $\omega$-regular objectives.
First, \emph{reachability objectives} (also known as reach-a-set objectives), the goal of which is to reach a target.
Second, \emph{safety objectives} (also known as stay-in-a-set objectives), the goal of which is to remain within a subset of vertices.
For general $\omega$-regular objectives, we consider \emph{Muller objectives}.
A Muller objective is described by a family of subsets of vertices; a play satisfies a Muller objective if the set of vertices occurring infinitely often is in the family defining the Muller objective.
Muller objectives generalise Büchi, CoBüchi and parity objectives.

\subparagraph*{When is memory needed?}
We investigate the existence of memoryless equilibria in games on turn-based deterministic arenas.
This question has been settled for more general arenas.
We briefly survey known results regarding the existence of (memoryless) equilibria in games on concurrent (stochastic) arenas and turn-based stochastic arenas.
In a concurrent arena, the players select actions simultaneously in each vertex and the next vertex is selected following a distribution depending only on the current vertex and chosen actions.
A turn-based stochastic arena can be seen as a concurrent arena in which only at most one player has more than one choice of action at each vertex.

For discounted concurrent stochastic games, memoryless randomised SPE are known to exist~\cite{Fink64,Takahashi64}.
This property does not generalise to concurrent games with reachability, safety or Muller objectives.
For safety objectives, while NEs exist in concurrent games~\cite{DBLP:journals/ijgt/SecchiS02}, there need not exist a memoryless NE even on turn-based stochastic arenas~\cite{DBLP:journals/corr/abs-1903-11935}.
For reachability objectives, NEs need not exist in the concurrent case~\cite{DBLP:journals/mor/FleschTV96}, although memoryless $\varepsilon$-NE have been shown to exist~\cite{DBLP:conf/csl/ChatterjeeMJ04} (i.e., strategy profiles in which no unilateral deviation increases the probability of the objective of a player by $\varepsilon$) for all $\varepsilon > 0$.
In turn-based stochastic games with $\omega$-regular objectives, pure NE are known to exist~\cite{DBLP:conf/csl/ChatterjeeMJ04}.
However, there need not exist (randomised) memoryless NE already in turn-based  reachability games with absorbing targets~\cite{DBLP:journals/eor/KuipersFSV09,DBLP:phd/dnb/Ummels11}.
This negative example for reachability can also be seen as a game with a Muller objective.
For games on turn-based deterministic arenas, we highlight the example in~\cite{DBLP:journals/dam/BorosGMOV18} of a turn-based deterministic game with no memoryless NE, where the goal of the players is to optimise some payoff attributed to some absorbing vertices and the unique cycle of the arena.
Finally, for complex Muller objectives (in the sense of the Mostowski hierarchy)\todo{Reference here?}, there need not exist memoryless randomised NEs in turn-based deterministic zero-sum games~\cite{DBLP:conf/qest/ChatterjeeAH04}.

All of the negative results above require either a stochastic game, payoff functions or complex Muller objectives.
It follows that they do not settle the problem of existence of memoryless equilibria in reachability and safety games, and in games with simple Muller objectives on turn-based deterministic arenas.

\subparagraph*{Our contribution.}
We prove the existence of \emph{memoryless randomised SPE} in games where all players have an objective from three classes: safety objectives, reachability objectives, or simple Muller objectives, where all winning sets of vertices are up- or downwards closed (such as parity objectives with colours 0, 1, and 2).
Our proof is constructive and can be used to effectively compute strategy profiles that are \emph{everywhere NEs (ENEs)} — strategy profiles that are NEs from every vertex of the game, rather than from a fixed initial vertex. This is strictly stronger than NEs from a single initial vertex. For this class of games, every ENE is an SPE, although this implication does not hold in general. Figure~\ref{fig:spe-not-ene} illustrates this distinction: the memoryless strategy profile in Figure~\ref {fig:spe-not-ene}b is an SPE from the designated initial vertex but fails to be an ENE, whereas the strategy profile depicted in Figure~\ref{fig:spe-not-ene}(c) is both an ENE and an SPE.
Moreover, building these equilibria is simple and elegant, and our algorithm runs in polynomial time.

We complement our existence results with negative examples, which show that we cannot expect much more: when we have players with reachability (or B\"uchi) and safety (or CoB\"uchi) objectives, there may not be positional SPEs.
For Muller objectives that do not fall into the class we treat, the situation is worse: even if we ask for the dual class, where all losing sets of vertices are up- or downwards closed, these games may not have any memoryless NE at all. 

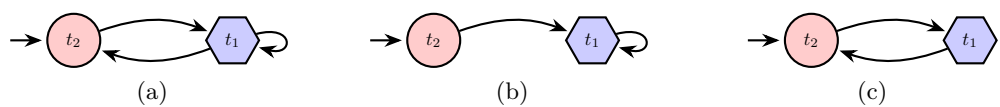
\begin{figure}[t]
\centering
\begin{minipage}{0.32\textwidth}
\centering
\begin{tikzpicture}[
    scale=0.7,
    transform shape,
    >=Stealth,
    shorten >=1pt,
    auto,
    every state/.style={draw, thick, minimum size=10mm, inner sep=1pt}
]
\tikzset{
    p1/.style={draw, circle, thick, minimum size=10mm, inner sep=1pt, fill=red!20},
    p2/.style={draw, regular polygon, regular polygon sides=6, thick, minimum size=10mm, inner sep=1pt, fill=blue!20}
}
\node[p1] (t2) at (0,0) {$t_2$};
\node[p2] (t1) at (3,0) {$t_1$};
\draw[->, thick] (-1.2,0) -- (-0.55,0);
\path[->, thick]
    (t2) edge[bend left=20] (t1)
    (t1) edge[loop right] (t1)
    (t1) edge[bend left=20] (t2);
\end{tikzpicture}

\small (a)

\end{minipage}
\hfill
\begin{minipage}{0.32\textwidth}
\centering
\begin{tikzpicture}[
    scale=0.7,
    transform shape,
    >=Stealth,
    shorten >=1pt,
    auto,
    every state/.style={draw, thick, minimum size=10mm, inner sep=1pt}
]
\tikzset{
    p1/.style={draw, circle, thick, minimum size=10mm, inner sep=1pt, fill=red!20},
    p2/.style={draw, regular polygon, regular polygon sides=6, thick, minimum size=10mm, inner sep=1pt, fill=blue!20}
}
\node[p1] (t2) at (0,0) {$t_2$};
\node[p2] (t1) at (3,0) {$t_1$};
\draw[->, thick] (-1.2,0) -- (-0.55,0);
\path[->, thick]
    (t2) edge[bend left=20] (t1)
    (t1) edge[loop right] (t1);
\end{tikzpicture}

\small (b)

\end{minipage}
\hfill
\begin{minipage}{0.32\textwidth}
\centering
\begin{tikzpicture}[
    scale=0.7,
    transform shape,
    >=Stealth,
    shorten >=1pt,
    auto,
    every state/.style={draw, thick, minimum size=10mm, inner sep=1pt}
]
\tikzset{
    p1/.style={draw, circle, thick, minimum size=10mm, inner sep=1pt, fill=red!20},
    p2/.style={draw, regular polygon, regular polygon sides=6, thick, minimum size=10mm, inner sep=1pt, fill=blue!20}
}
\node[p1] (t2) at (0,0) {$t_2$};
\node[p2] (t1) at (3,0) {$t_1$};
\draw[->, thick] (-1.2,0) -- (-0.55,0);
\path[->, thick]
    (t2) edge[bend left=20] (t1)
(t1) edge[bend left=20] (t2);
\end{tikzpicture}

\small (c)

\end{minipage}
\caption{A two-player reachability game. Circles represent $\Player_1$ and hexagons represent $\Player_2$. Vertex $ti$ is the target of $\Player_i$ (b) illustrates a memoryless strategy profile that is an SPE from $t_2$ but not an ENE. (c) illustrates a memoryless strategy profile, which is both an ENE and an SPE.}
\label{fig:spe-not-ene}
\end{figure}

\subparagraph*{Related work.}
We discuss three related directions: memory requirements in games, decision problems for equilibria and the interplay of memory and randomisation in strategy complexity.

Memory requirements in games on graphs have been thoroughly studied for zero-sum games.
In turn-based deterministic arenas, Gimbert and Zielonka provide a characterisation of objectives for which there exist optimal memoryless strategies for both players~\cite{GZ05}.
In the finite-memory case, there exists an analogous characterisation for objectives for which the memory structure for winning strategies of both players can be chosen independently of the arena~\cite{DBLP:journals/lmcs/BouyerLORV22}.
Characterisations of objectives for which memoryless strategies and finite-memory strategies suffice to win are given in~\cite{DBLP:journals/theoretics/Ohlmann23} and~\cite{DBLP:journals/lmcs/CasaresO25} respectively using universal graphs.
For multiplayer games,~\cite{DBLP:journals/iandc/RouxP18} provides sufficient conditions for the existence of finite-memory NEs based on the existence of finite-memory winning strategies in related zero-sum games.

From an algorithmic perspective, deciding the existence of an NE or an SPE is trivial in $\omega$-regular games on turn-based deterministic arenas, as they always exist.
A more relevant problem is the constrained equilibrium existence problem.
In pure strategies, it asks whether there exists an equilibrium from a given initial vertex whose unique outcome is winning for (at least) at given subset of players.
In randomised strategies, this translates to lower bounds on the satisfaction probability of each objective by the equilibrium.
The main motivation is that, in general, one is not interested in any equilibrium, but instead desires an equilibrium that satisfies some objectives.
The constrained (pure) NE existence problem is \textsf{NP}-complete for reachability~\cite{DBLP:journals/jcss/BrihayeBGT21}, safety~\cite{DBLP:conf/icalp/ConduracheFGR16} and co-Büchi~\cite{DBLP:conf/fossacs/Ummels08} games, and it can be solved in polynomial time for Büchi games.
The constrained (pure) SPE existence problem is \textsf{PSPACE}-complete for reachability and safety games~\cite{DBLP:journals/iandc/BrihayeBGR21}.
Solutions to the constrained equilibrium existence problem require memory in general; see~\cite{DBLP:conf/fsttcs/BrihayeGMR23} for a discussion of the role of memory in solutions to the constrained equilibrium existence problem in reachability games and~\cite{DBLP:journals/iandc/Main26} for upper bounds on memory requirements for solutions to the constrained NE existence problem in reachability, safety, Büchi and co-Büchi games.
The constrained equilibrium existence problem has also been studied in more general models, e.g., in concurrent games~\cite{DBLP:conf/fsttcs/BouyerBMU11,DBLP:phd/hal/Brenguier12} and timed games~\cite{DBLP:conf/concur/BouyerBM10,DBLP:conf/formats/BrihayeG20}.

Finally, we discuss the interplay of memory and randomisation.
In our construction, we use randomisation to ensure that all vertices of some parts of the arena are visited infinitely often instead of doing so with memory.
Several works study trade-offs between memory and randomisation.
For instance, it is shown in~\cite{DBLP:conf/qest/ChatterjeeAH04} that in zero-sum Muller games with upward-closed accepting sets, memory is necessary in general for pure winning strategies but not when using randomised strategies.
The same property holds in games where the goal is to enforce a conjunction of a parity and mean-payoff objective~\cite{DBLP:journals/acta/ChatterjeeRR14}.
More generally, one can lessen the amount of memory needed by considering randomised finite-memory strategies; see, e.g.,~\cite{DBLP:conf/fossacs/Chatterjee07,DBLP:conf/stacs/Horn09} for zero-sum Muller games.
Memory bounds are sensitive to the chosen definition of a randomised finite-memory strategy; not all models are equally compact, nor are they equally expressive.
The relative expressiveness of these models is discussed in~\cite{DBLP:journals/iandc/MainR24}.

\subparagraph*{Outline.} 
We provide required definitions in Section~\ref{sec:prelim}.
In Section~\ref{sec:neg}, we present examples that establish the need of randomisation to obtain memoryless SPEs in the class of games we study, and show that memoryless NEs need not exist for Muller objectives outside of the class we treat.
In Section~\ref{sec:games}, we prove our existence result: we first illustrate our construction then present an algorithm to construct memoryless SPEs. 
\section{Preliminaries}\label{sec:prelim}
\subparagraph*{Notation.}
For any $n\geq 1$, we let $\llb n \rrb = \{1, \ldots, n\}$.
For any finite set $A\neq\emptyset$, we let $\mathcal{D}(A)$ denote the set of distributions over $A$, i.e., of functions $\mu\colon A\to [0, 1]$ such that $\sum_{a\in A}\mu(a) = 1$.

\subparagraph*{Arenas.}
Let $n \geq 1$. For each $i \in \llb n \rrb$, we denote by $\Player_i$ the $i$-th player. Let $G = (V,E)$ be a directed graph where $V$ is a finite set of vertices and $E \subseteq V \times V$ is a set of edges. We assume that the graph is deadlock-free, i.e., for all $v \in V$, there exists $v' \in V$ such that $(v,v') \in E$. An $n$-player arena on $G$ is a tuple $A = ((V_i)_{i \in \llb n \rrb}, E)$, where $(V_i)_{i \in \llb n \rrb}$ is a partition of $V$. 

Players take turns choosing the next vertex depending on which player controls the current vertex. This induces an infinite path in the arena, called a play. Formally, a play in $A$ is an infinite sequence of vertices $\rho = \rho_0 \rho_1 \dots\in V^\omega$ such that for all $\ell \in \mathbb{N}$, $(\rho_\ell, \rho_{\ell+1}) \in E$. 
A history is a finite prefix of a play $h = h_0 \dots h_\ell \in V^*$; we denote $h_\ell$ by $\mathrm{last}(h)$. The set of plays (resp. histories) is denoted by $\mathrm{Plays}$ (resp. $\mathrm{Hist}$), and we let $\mathrm{Hist}_i=\mathrm{Hist}\cap V^*V_i$.

Let $C$ be a strongly connected component (SCC) of $G$. We say that $C$ is a \emph{leaf} SCC if there are no edges from $C$ to $V \setminus C$, i.e., once a play enters $C$, it remains in $C$ forever.

\subparagraph*{Strategies.}
Strategies describe how players select edges.
They can use both memory and randomisation in general.
Let $i \in \llb n \rrb$.
Formally, a (randomised) strategy of $\Player_i$ is a function $\sigma_i : \mathrm{Hist}_i \to \mathcal{D}(V)$ such that for all $h\in \mathrm{Hist}_i$ and all $v\in V$, if $\sigma_i(h)(v) > 0$, then $(\mathrm{last}(h), v) \in E$.
A strategy $\sigma_i$ is pure if it uses no randomisation, i.e., if for all $h \in\mathrm{Hist}_i$, $\sigma_i(h)$ is a Dirac distribution.
A strategy $\sigma_i$ is memoryless if for all histories $h$, $h'\in\mathrm{Hist}_i$, $\mathrm{last}(h) = \mathrm{last}(h')$ implies $\sigma_i(h) = \sigma_i(h')$.
We view memoryless strategies as functions $V_i\to\mathcal{D}(V)$.
Strategies that are both memoryless and pure are called positional.
We say that a play $\rho$ is played according to $\sigma_i$ if, whenever $\rho_\ell \in V_i$, we have that $\sigma_i(\rho_0 \dots \rho_\ell)(\rho_{\ell+1}) > 0$.

A strategy profile is a tuple $\sigma = (\sigma_1, \ldots, \sigma_n)$ where, for all $i\in\llb n\rrb$, $\sigma_i$ is a strategy of $\Player_i$.
We write $\mathbb{P}_v^\sigma$ for the distribution over plays induced by a strategy profile $\sigma=(\sigma_1,\dots,\sigma_n)$ from an initial vertex $v\in V$, defined in the usual way (see, e.g.,~\cite{DBLP:journals/iandc/MainR24}).
A strategy profile $\sigma$ is called memoryless (resp.~positional) if all strategies in $\sigma$ are memoryless (resp.~positional).

\subparagraph*{Objectives and games.}
An objective is a measurable subset of $\mathrm{Plays}$.
We say that a play $\rho$ satisfies an objective $\Omega$ if $\rho\in\Omega$.
An $n$-player game $\G = (A, (\Omega_i)_{i \in \llb n \rrb})$ is an arena equipped with a profile of objectives $(\Omega_i)_{i \in \llb n \rrb}$.

We focus on the following three classes of objectives.
 A reachability objective is defined by a target set $T \subseteq V$; a play $\rho$ satisfies the reachability objective for $T$ if it visits $T$.
 Formally, we let $\Reach(T) = \{\rho\in\mathrm{Plays}\mid \exists \ell\geq 0, \rho_\ell\in T\}$.
 A safety objective is defined by a set $S \subseteq V$; a play $\rho$ satisfies the safety objective for $S$ if it never leaves $S$.
 Formally, we let $\Safe(S) = \{\rho\in\mathrm{Plays}\mid \forall\ell\geq 0, \rho_\ell\in S\}$.
 
 For any play $\rho$, let $\mathrm{Inf}(\rho)$ be the set of vertices occurring infinitely often in $\rho$. A Muller objective is defined by a family $\mathcal{F} \subseteq 2^V$. A play $\rho$ satisfies the Muller objective described by $\mathcal{F}$, denoted by $\mathrm{Muller}(\mathcal{F})$, if $\mathrm{Inf}(\rho) \in \mathcal{F}$.
 An objective $\mathrm{Muller}(\mathcal{F})$ is a 0-2 Muller objective if for all $F\in\mathcal{F}$, either all subsets of $F$ are in $\mathcal{F}$ ($F$ is downward closed in $\mathcal{F}$) or all supersets of $F$ are in $\mathcal{F}$ ($F$ is upward closed in $\mathcal{F}$), and $\mathrm{Muller}(\mathcal{F})$ is a 1-3 Muller objective if its complement $\mathrm{Muller}(2^V\setminus\mathcal{F})$ is a 0-2 Muller objective.\todo{James: @Mona, please double check.}
 
\subparagraph*{Nash equilibria.}
Let $\G = (A, (\Omega_i)_{i \in \llb n \rrb})$ be a game and $v\in V$ and let $\sigma$ be a strategy profile. A strategy $\sigma_i'$ of $\Player_i$ is a profitable deviation for $\Player_i$ from $v$ with respect to $\sigma$ if ${\mathbb{P}_{v}^{(\sigma_{-i},\sigma_i')}}(\Omega_i) > {\mathbb{P}_{v}^\sigma}(\Omega_i)$, where $(\sigma_{-i},\sigma_i')$ denotes the profile obtained by replacing the strategy of $\Player_i$ in $\sigma$ by $\sigma_i'$. 
A strategy profile $\sigma$ is a Nash equilibrium (NE) from $v\in V$ if there are no profitable deviations with respect to $\sigma$ from $v$. A strategy profile $\sigma$ is an everywhere NE (ENE) if for every vertex $v \in V$, $\sigma$ is an NE from $v$.

\subparagraph*{Subgame-perfect equilibria.}
Given a possibly empty history $h\in\{\varepsilon\}\cup\mathrm{Hist}$, we define the subgame of $\G$ rooted at $h$ as $\G_h = (A, (\Omega^i_h)_{i \in \llb n \rrb})$, where for all $i \in \llb n \rrb$, we have $\Omega^i_h = \{\rho\in\mathrm{Plays}\mid h\rho\in\Omega^i\}$.
Given a strategy profile $\sigma$ and a history $h$, we let $\sigma_h$ denote the strategy profile defined by $(\tau_1, \ldots, \tau_n)$ such that, for all $i\in\llb n\rrb$ and all histories $h'\in\mathrm{Hist}_i$ such that $hh'$ is a history, $\tau_i(h') = \sigma_i(hh')$.
A strategy profile $\sigma$ is a subgame-perfect equilibrium (SPE) from $v\in V$ in $\G$ if, for all histories $hv'\in\mathrm{Hist}$ starting in $v$, $\sigma^h$ is an NE in $\G_h$ from $v'$. In particular, every SPE is an NE.
 
\section{Impossibility Results}
\label{sec:neg}

Before we introduce the construction of randomised SPEs for any combination of players with reachability, safety, or 0-2 Muller objective, we first show that randomisation is required even if we only consider reachability and safety 
objectives, while 1-3 Muller objectives may not have any memoryless NEs at all.

\subsection{Randomisation}
We illustrate the need of randomisation for reachability (or B\"uchi) and safety (or CoB\"uchi) 
on a simple example.

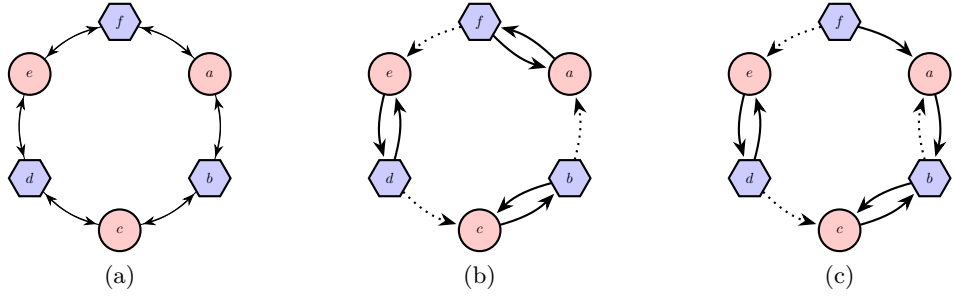
\begin{figure}[t]
\centering

\begin{minipage}{0.32\textwidth}
\centering
\begin{tikzpicture}[
    scale=0.55,
    transform shape,
    >=Stealth,
    shorten >=1pt,
    auto,
    every state/.style={draw, thick, minimum size=8mm, inner sep=1pt}
]

\tikzset{
    p1/.style={draw, circle, thick, minimum size=10mm, inner sep=1pt, fill=red!20},
    p2/.style={draw, regular polygon, regular polygon sides=6, thick, minimum size=10mm, inner sep=1pt, fill=blue!20}
}

\node[p1] (a) at (30:2.5cm) {$a$};
\node[p2] (b) at (-30:2.5cm) {$b$};
\node[p1] (c) at (270:2.5cm) {$c$};
\node[p2] (d) at (210:2.5cm) {$d$};
\node[p1] (e) at (150:2.5cm) {$e$};
\node[p2] (f) at (90:2.5cm) {$f$};

\path[->, ]

(a) edge[bend left=15] (b)
(b) edge[bend right=15] (a)

(b) edge[bend left=15] (c)
(c) edge[bend right=15] (b)

(c) edge[bend left=15] (d)
(d) edge[bend right=15] (c)

(d) edge[bend left=15] (e)
(e) edge[bend right=15] (d)

(e) edge[bend left=15] (f)
(f) edge[bend right=15] (e)

(f) edge[bend left=15] (a)
(a) edge[bend right=15] (f);

\end{tikzpicture}

\small (a)
\end{minipage}
\hfill
\begin{minipage}{0.32\textwidth}
\centering
\begin{tikzpicture}[
    scale=0.55,
    transform shape,
    >=Stealth,
    shorten >=1pt,
    auto,
    every state/.style={draw, thick, minimum size=8mm, inner sep=1pt}
]

\tikzset{
    p1/.style={draw, circle, thick, minimum size=10mm, inner sep=1pt, fill=red!20},
    p2/.style={draw, regular polygon, regular polygon sides=6, thick, minimum size=10mm, inner sep=1pt, fill=blue!20}
}

\node[p1] (a) at (30:2.5cm) {$a$};
\node[p2] (b) at (-30:2.5cm) {$b$};
\node[p1] (c) at (270:2.5cm) {$c$};
\node[p2] (d) at (210:2.5cm) {$d$};
\node[p1] (e) at (150:2.5cm) {$e$};
\node[p2] (f) at (90:2.5cm) {$f$};

\path[->, thick]

(a) edge[bend right=15] (f)

(b) edge[bend right=15, dotted] (a)
(b) edge[bend right=15] (c)

(c) edge[bend right=15] (b)

(d) edge[bend right=15, dotted] (c)
(d) edge[bend right=15] (e)

(e) edge[bend right=15] (d)

(f) edge[bend right=15, dotted] (e)
(f) edge[bend right=15] (a);
\end{tikzpicture}

\small (b)
\end{minipage}
\hfill
\begin{minipage}{0.32\textwidth}
\centering
\begin{tikzpicture}[
    scale=0.55,
    transform shape,
    >=Stealth,
    shorten >=1pt,
    auto,
    every state/.style={draw, thick, minimum size=8mm, inner sep=1pt}
]

\tikzset{
    p1/.style={draw, circle, thick, minimum size=10mm, inner sep=1pt, fill=red!20},
    p2/.style={draw, regular polygon, regular polygon sides=6, thick, minimum size=10mm, inner sep=1pt, fill=blue!20}
}

\node[p1] (a) at (30:2.5cm) {$a$};
\node[p2] (b) at (-30:2.5cm) {$b$};
\node[p1] (c) at (270:2.5cm) {$c$};
\node[p2] (d) at (210:2.5cm) {$d$};
\node[p1] (e) at (150:2.5cm) {$e$};
\node[p2] (f) at (90:2.5cm) {$f$};

\path[->, thick]

(a) edge[bend left=15] (b)

(b) edge[bend left=15, dotted] (a)
(b) edge[bend right=15] (c)

(c) edge[bend right=15] (b)

(d) edge[bend right=15, dotted] (c)
(d) edge[bend right=15] (e)

(e) edge[bend right=15] (d)

(f) edge[bend right=15, dotted] (e)
(f) edge[bend left=15] (a);
\end{tikzpicture}

\small (c)
\end{minipage}

\caption{A six-player game: circles denote reachability players, hexagons safety players. (a) Game graph. (b) Fixing the direction of the play for reachability players. (c) The play unstable under unilateral deviations (Case (1), of Example~\ref{ex:mixobj}).}

\label{fig:case1}
\end{figure}

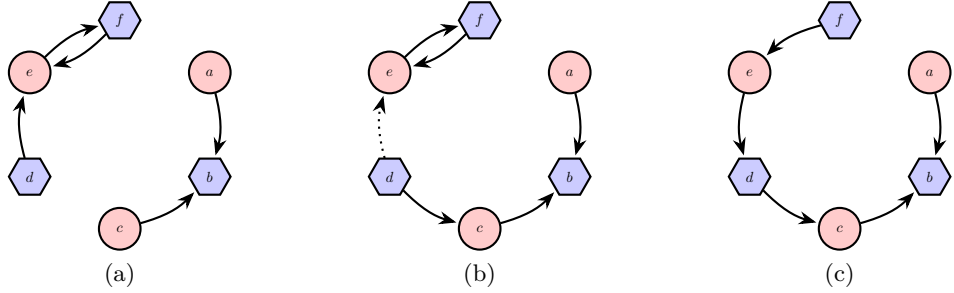
\begin{figure}[t]
\centering

\begin{minipage}{0.32\textwidth}
\centering
\begin{tikzpicture}[
    scale=0.55,
    transform shape,
    >=Stealth,
    shorten >=1pt,
    auto,
    every state/.style={draw, thick, minimum size=8mm, inner sep=1pt}
]

\tikzset{
    p1/.style={draw, circle, thick, minimum size=10mm, inner sep=1pt, fill=red!20},
    p2/.style={draw, regular polygon, regular polygon sides=6, thick, minimum size=10mm, inner sep=1pt, fill=blue!20}
}

\node[p1] (a) at (30:2.5cm) {$a$};
\node[p2] (b) at (-30:2.5cm) {$b$};
\node[p1] (c) at (270:2.5cm) {$c$};
\node[p2] (d) at (210:2.5cm) {$d$};
\node[p1] (e) at (150:2.5cm) {$e$};
\node[p2] (f) at (90:2.5cm) {$f$};

\path[->, thick]

(a) edge[bend left=15] (b)

(c) edge[bend right=15] (b)

(d) edge[bend left=15] (e)

(e) edge[bend left=15] (f)
(f) edge[bend left=15] (e);

\end{tikzpicture}

\small (a)
\end{minipage}
\hfill
\begin{minipage}{0.32\textwidth}
\centering
\begin{tikzpicture}[
    scale=0.55,
    transform shape,
    >=Stealth,
    shorten >=1pt,
    auto,
    every state/.style={draw, thick, minimum size=8mm, inner sep=1pt}
]

\tikzset{
    p1/.style={draw, circle, thick, minimum size=10mm, inner sep=1pt, fill=red!20},
    p2/.style={draw, regular polygon, regular polygon sides=6, thick, minimum size=10mm, inner sep=1pt, fill=blue!20}
}

\node[p1] (a) at (30:2.5cm) {$a$};
\node[p2] (b) at (-30:2.5cm) {$b$};
\node[p1] (c) at (270:2.5cm) {$c$};
\node[p2] (d) at (210:2.5cm) {$d$};
\node[p1] (e) at (150:2.5cm) {$e$};
\node[p2] (f) at (90:2.5cm) {$f$};

\path[->, thick]
(a) edge[bend left=15] (b)

(c) edge[bend right=15] (b)

(d) edge[bend left=15] [dotted] (e)
(d) edge[bend right=15] (c)

(e) edge[bend left=15] (f)
(f) edge[bend left=15] (e);
\end{tikzpicture}

\small (b)
\end{minipage}
\hfill
\begin{minipage}{0.32\textwidth}
\centering
\begin{tikzpicture}[
    scale=0.55,
    transform shape,
    >=Stealth,
    shorten >=1pt,
    auto,
    every state/.style={draw, thick, minimum size=8mm, inner sep=1pt}
]

\tikzset{
    p1/.style={draw, circle, thick, minimum size=10mm, inner sep=1pt, fill=red!20},
    p2/.style={draw, regular polygon, regular polygon sides=6, thick, minimum size=10mm, inner sep=1pt, fill=blue!20}
}

\node[p1] (a) at (30:2.5cm) {$a$};
\node[p2] (b) at (-30:2.5cm) {$b$};
\node[p1] (c) at (270:2.5cm) {$c$};
\node[p2] (d) at (210:2.5cm) {$d$};
\node[p1] (e) at (150:2.5cm) {$e$};
\node[p2] (f) at (90:2.5cm) {$f$};

\path[->, thick]
(a) edge[bend left=15] (b)

(c) edge[bend right=15] (b)

(d) edge[bend right=15] (c)

(e) edge[bend right=15] (d)
(f) edge[bend right=15] (e);
\end{tikzpicture}

\small (c)
\end{minipage}

\caption{Case (2) of Example~\ref{ex:mixobj}: fixing $f$ counterclockwise forces $d$ to play clockwise, yielding an unstable profile under unilateral deviations.}
\label{fig:case2}
\end{figure}

\begin{example}
\label{ex:mixobj}
We now consider a game where six players are standing in a circle, and in each move they hand the token to a neighbour (see Fig.~\ref{fig:case1}a). 
Players $a$, $c$, and $e$ have the individual reachability objective that the token shall leave their respective neighbourhood (their own nodes + the two neighbouring nodes), or the B\"uchi objective that it shall do so infinitely often.
The other three players, $b$, $d$, and $f$, have the individual safety objective never to leave their respective neighbourhood, or the CoB\"uchi objective to eventually always stay in this neighbourhood.

We now examine whether a positional ENE exists for this example.
We start with the strategies of the reachability/B\"uchi players and observe that, modulo renaming, there are only two such strategies:

\begin{enumerate}
    \item All reachability/B\"uchi
    players pass the token in the same direction (w.l.o.g.\ counter-clockwise).

Figure \ref{fig:case1}b shows this case together with the only optimal counter-strategy of the safety/CoB\"uchi players to always play back/clockwise.
(All other positional strategies of the individual safety/CoB\"uchi 
players can be improved upon and are therefore not stable).
But for this strategy of the safety/CoB\"uchi 
players, none of the strategies of the reachability/B\"uchi players is stable: they lose from their vertices, but would win if they played clockwise. (violates the NE condition).

\item Two reachability/B\"uchi 
players pass the token toward the same player; w.l.o.g.\ we consider the case where $a$ and $c$ play to $b$, while $e$ plays to $f$. (All other such strategies are equivalent to this one modulo renaming.)

Player $f$ now has only one stable policy (to move back to $e$), and in an ENE, it has to take it.
We therefore fix this decision.

If $d$ also plays to $e$ (Figure \ref{fig:case2}a), then the respective reachability/B\"uchi player can win from both $a$ and $c$ by changing their policy.
But for either positional strategy played from $b$, either the reachability/B\"uchi player owning $a$ or the reachability/B\"uchi player owning $c$ loses from that vertex.
Her choice is therefore not stable, so this cannot happen in an ENE.

If $d$ instead plays to $c$, (Figure \ref{fig:case2}b), then $e$ can improve her outcome by turning to $d$ instead (Figure \ref{fig:case2}c), so this also cannot result in an ENE.
\end{enumerate} 
\end{example}

We can extend the example from ENEs to SPEs by simply adding a seventh vertex for a seventh player.
This vertex is initial, safe for all safety players and not a target for either reachability player, and has all six vertices as successors, then the same argument works---it is just more difficult to draw and the additional initial vertex distracts from the core argument.

\begin{theorem}\label{thm:noPure}
Games with reachability or B\"uchi objectives as well as safety or CoB\"uchi objectives may neither allow for positional ENE nor for positional SPE.
\end{theorem}

\subsection{1-3 Muller games}
We now show that Muller games where all rejecting infinity sets are upwards \emph{or} downwards closed may not have \emph{any} memoryless NE.

\begin{figure}[t]
\centering
\begin{tikzpicture}[scale=0.8]

\fill[red!10, opacity=0.5]
(-3,4) -- (-3,2.8)
.. controls (0,2.5) .. (3,2)
-- (3,4) -- cycle;

\node at (0,3.65) {\small memoryless NE may not exist};

\fill[green!10, opacity=0.5]
(-3,-1) -- (-3,2.8)
.. controls (0,2.5) .. (3,2)
-- (3,-1) -- cycle;

\node at (0,-0.6) {\small memoryless NE always exist};

\node (g) at (0,3.2) {$\vdots$};
\node (e) at (-2,2.4) {$\{0,1,2\}$};
\node (f) at ( 2,2.4) {$\{1,2,3\}$};
\node (d) at ( 0,1.6) {$\{0,1,2\} \cap \{1,2,3\}$};

\node (b) at (-2,0.8) {$\{0,1\}$};
\node (c) at ( 2,0.8) {$\{1,2\}$};
\node (a) at ( 0,0) {$\{0,1\}\cap\{1,2\}$};

\draw (d) -- (f);
\draw (d) -- (e);

\draw (d) -- (f);
\draw (d) -- (e);

\draw (f) -- (g);
\draw (e) -- (g);

\draw (b) -- (d);
\draw (c) -- (d);

\draw (a) -- (c);
\draw (a) -- (b);

\draw[dashed, thick, gray] (-2.5,2.8) .. controls (0,2.5) .. (2.5,2);

\end{tikzpicture}
\caption{The Mostowski hierarchy. For all levels below the separator, there exists a memoryless (randomised) ENE and SPE.}
\end{figure}
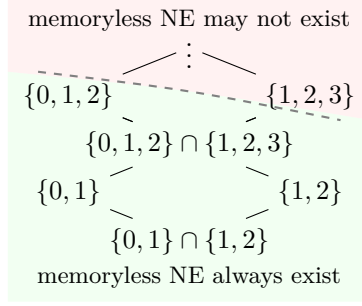
\label{fig:mostowski}

For this, we first very briefly outline the Mostowski hierarchy~\cite{DBLP:conf/sct/Mostowski84} for Muller objectives (see Figure~\ref{fig:mostowski}).
The 0-2 class is those objectives, where all winning sets are up- or downward closed, such as parity objectives with colours 0,1,2.
The 1-3 class is the reverse: all losing sets are up- or downward closed, such as parity objectives with colours 1,2,3.

While we will treat the 0-2 class in the algorithm we introduce in the next section, we now provide an example, where two players with Muller objectives from the 1-3 class do not have any memoryless NE.

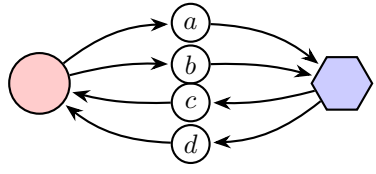
\begin{figure}[t]
\centering

\begin{tikzpicture}[
    >=Stealth,
    shorten >=1pt,
    thick
]
\tikzset{
    p1/.style={draw, circle, thick, minimum size=8mm, inner sep=1pt, fill=red!20},
    p2/.style={draw, regular polygon, regular polygon sides=6, thick, minimum size=8mm, inner sep=1pt, fill=blue!20},
    mid/.style={draw, circle, minimum size=5mm, inner sep=1pt}
}
\node[p1] (1) at (0,0) {};
\node[p2] (2) at (4,0) {};

\node[mid] (a) at (2,0.8) {$a$};
\node[mid] (b) at (2,0.25) {$b$};
\node[mid] (c) at (2,-0.25) {$c$};
\node[mid] (d) at (2,-0.8) {$d$};

\path[->]
(1) edge[bend left=18] (a)
(1) edge[bend left=8]  (b)
(a) edge[bend left=18] (2)
(b) edge[bend left=8]  (2)

(2) edge[bend left=8]  (c)
(2) edge[bend left=18] (d)
(c) edge[bend left=8]  (1)
(d) edge[bend left=18] (1);

\end{tikzpicture}

\caption{A 1-3 Muller game with no memoryless NEs. Circles represent $\Player_1$ and hexagons represent $\Player_2$.}
\label{fig:muller-example}
\end{figure}

\begin{example}

Consider a game with two players (circle and hexagon), where each player has two outgoing transitions to the other vertex, as depicted in Figure~\ref{fig:muller-example}. The players have almost opposing objectives, except that both players lose if all four middle vertices are seen infinitely often.
If three are seen infinitely often, the player for whom only one of her successor vertices is seen infinitely often wins.
If two of the middle vertices are seen infinitely often, 
$\Player_1$ wins for vertices $a$,$c$ and $b$,$d$, while $\Player_2$ wins for vertices $a$,$d$ and $b$,$c$.

This way, if both players use randomisation, then both benefit from deviating, so this is no NE.
If both players play pure, the losing player can win by changing to her other pure choice, so this cannot be an NE.
If one player plays randomised while the other plays pure, the randomised player loses, but can change to a winning pure choice, so this cannot be an NE either.
\end{example}

We note that the losing sets with three and four of the middle vertices are upwards closed for either player,
while the losing sets with two vertices are downwards closed for either player, so that these are 1-3 Muller objectives.

\begin{theorem}\label{thm:1-3}
Games with 1-3 Muller objectives may not have NEs, even if we restrict the games to mutually exclusive 1-3 Muller objectives and two players.
\end{theorem}

\section{Memoryless Nash Equilibria in Games with Reachability, Safety, and 0-2 Muller Objectives}
In this section, we study the existence of memoryless NE for games with reachability, safety, and 0-2 Muller objectives. We first illustrate the principles and ideas of our construction for memoryless randomised NE for reachability, safety, and 0-2 Muller objectives. 
We then give the proof of the main result in Section~\ref{sec:algo}.

\label{sec:games}
\newcommand{\jamesColor}[1]{\textcolor{black}{#1}}
\newcommand{\monaColor}[1]{\textcolor{blue}{#1}}

\subsection{Reachability Objectives}

\begin{example}\label{ex:reach}
Consider the four-player game with reachability objectives played on the arena of Figure~\ref{fig:reachability-example}. 
The main idea is to incrementally \jamesColor{define memoryless strategies in a way that allows as many players as possible to win}.
A vertex is considered evaluated once a move is assigned to it.
\jamesColor{When evaluating a vertex, we record} the set of players who will win when the game starts at this vertex.

For reachability games, \jamesColor{desirable strategies} are particularly easy to identify in leaf SCCs: a \emph{maximally mixed} strategy (i.e., such that all outgoing transitions of every vertex are played uniformly at random) will almost-surely visit all vertices in the leaf component infinitely often, thus in particular \jamesColor{almost-surely} satisfy any reachability objective with a target in the leaf.
At the same time, players with no target in the leaf component cannot win at all.

Let $C$ be the leaf SCC in the green area of Figure~\ref{fig:reachability-example}. 
Playing uniformly at random constitutes an NE from every vertex in $C$, we fix this strategy and can mark $t_2$, $v_3$, and $t_4$ as winning for $\Player_2$ and $\Player_4$, but not for the other two players.

Once we have evaluated vertices, we incrementally construct a strategy profile.
As part of this process, we remove edges to evaluated vertices that are not marked as winning for the player controlling the outgoing vertex.
Intuitively, such edges are not interesting for this player. 
To keep track of these modifications, we maintain a reduced arena containing exactly the edges that have not been removed.

In order to do this, we consider, for each unevaluated vertex $v \in V_i$ its successors.
\begin{enumerate}
    \item \label{case:reach:1}
If there exists a successor $v'$ of $v$ that is marked as winning for $\Player_i$, we fix the edge $v \to v'$ \jamesColor{in the strategy we build}.
This is safe to do because we have already established that $\Player_i$ almost-surely wins from $v'$, thus $\Player_i$ will almost-surely win from $v$ this way; no other successor can provide a better result, and the decision to move to $v$ will therefore be stable.
We then mark $v$ as winning for all players who are winning from $v'$, as well as for those whose target is $v$. For example, consider $v_5$ owned by $\Player_4$, with successors $v_3$ and $t_4$, both are marked as winning for her.
We fix one of them (here, $v_5 \to t_4$) and mark $v_5$ as winning for $\Player_2$ and $\Player_4$. 

\item\label{case:reach:2}
If $v$ has more than one successor, and among them is an evaluated successor $v'$ that is marked as losing for $\Player_i$, we remove the edge from $v \to v'$.
This is safe to do because we have already established that we lose from that successor; so whichever of the remaining successors will be chosen, moving instead to $v'$ cannot be advantageous for $\Player_i$; not considering this edge in the future cannot lead to instability.

For example, consider $v_2$ owned by $\Player_3$, with successors $v_1$, $v_3$, and $v_5$. Since $\Player_3$ does not win from the evaluated vertices $v_3$ and $v_5$, we can successively remove the edges $v_2 \to v_3$ and $v_2 \to v_5$. This leaves $v_2$ with a single successor.
While we cannot evaluate $v_2$ yet (as $v_1$ is not evaluated), we have simplified the game.

\item \label{case:reach:3}
If $v$ has only one successor $v'$ and that successor is evaluated, we fix the edge $v\to v'$ and evaluate $v$ as in Rule~\ref{case:reach:1}.
This is safe to do as there is no other edge left to consider.
\end{enumerate}

To continue with the example, we could next 
evaluate $t_1$, owned by $\Player_2$, fix the edge to $t_4$ (i.e., $t_1 \to t_4$), and mark $t_1$ as winning for $\Player_1$, $\Player_2$, and $\Player_4$ by Rule~\ref{case:reach:1}.
We continue by evaluating $t_3$, $v_1$, $v_2$, $v_4$, and $v_0$. The resulting strategy profile is depicted in Figure~\ref{fig:reachability-example}, where solid arrows represent the selected edges \jamesColor{(with associated probabilities)}. 

\jamesColor{The vertices may also be evaluated in a different order, i.e., the construction is nondeterministic.
For instance, Rule~\ref{case:reach:1} may offer different successor vertices to choose from while Rule~\ref{case:reach:2} may offer different edges to drop.
Therefore, the approach may yield several strategy profiles}.

If unevaluated vertices remain and none of Rules~\ref{case:reach:1}-\ref{case:reach:3} apply, all unevaluated vertices have only unevaluated successors.
Due to the edge deletions, there is
\jamesColor{a new leaf SCC consisting of unevaluated vertices}; we extend the strategy by letting the players play uniformly at random within this leaf SCC.
This phenomenon occurs in the game of Figure~\ref{fig:reachability-random}: after evaluating the leaf component $\{v_3, t_1, t_4\}$ and $v_5$, and removing the edges from $v_2$ and $t_1$ to the aforementioned vertices (their targets cannot be reached using them), a new leaf emerges in the reduced arena (Figure~\ref{fig:reachability-reduced}).
\end{example}

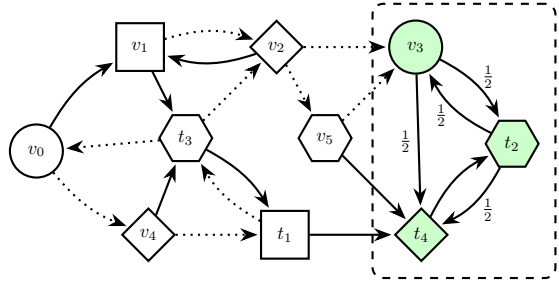
\begin{figure}[t]
\centering
\begin{tikzpicture}[
    scale=0.7,
    transform shape,
    >=Stealth,
    shorten >=1pt,
    auto,
    node distance=1.6cm,
    every state/.style={draw, thick, minimum size=8mm, inner sep=1pt}
]

\tikzset{
    p1/.style={draw, circle, thick, minimum size=10mm, inner sep=1pt},
    p2/.style={draw, rectangle, thick, minimum size=9mm, inner sep=1pt},
    p3/.style={draw, diamond, thick, minimum size=10mm, inner sep=1pt},
    p4/.style={draw, regular polygon, regular polygon sides=6, thick, minimum size=10mm, inner sep=1pt}
}

\node[p1] (0) {$v_0$};
\node[p2, above right=of 0] (1) {$v_1$};
\node[p3, right=of 1] (2) {$v_2$};
\node[p1, fill = green!20, right=of 2] (3) {$v_3$};

\node[p4, fill = green!20, below right=of 3] (9) {$t_2$};
\node[p3, fill = green!20,  below left=of 9] (8) {$t_4$};

\node[p2,left=of 8] (7) {$t_1$};
\node[p3, left=of 7] (6) {$v_4$};

\node[p4, below left=of 2] (4) {$t_3$};
\node[p4, right=of 4] (5) {$v_5$};

\node[draw, dashed, thick, rounded corners, fit=(3)(8)(9), inner sep=6pt] {};

\path[->, thick]
    (0) edge[bend left=15] (1)
    (0) edge[bend right=15] [dotted] (6)

    (1) edge (4)
    (1) edge[bend left=18] [dotted](2)
    
    (2) edge[bend left=18]  (1)
    (2) edge[dotted] (3)
    (2) edge[dotted] (5)

    (3) edge node[midway,left]{$\frac{1}{2}$} (8)
    (3) edge [bend left=18] node[midway ,above=4pt,right]{$\frac{1}{2}$} (9)

    (4) edge [dotted](0)
    (4) edge [dotted](2)
    (4) edge [bend left=15] (7)

    (5) edge (8)
    (5) edge[dotted] (3)

    (6) edge [dotted](7)
    (6) edge (4)

    (7) edge[bend left=15][dotted] (4)
    (7) edge (8)

    (8) edge[bend left=18] (9)

    (9) edge [bend left=18] node[midway ,below=4pt,left]{$\frac{1}{2}$} (3)
    (9) edge[bend left=18] node[midway ,below=4pt,right]{$\frac{1}{2}$} (8);

\end{tikzpicture}
\caption{A four-player reachability game where $\Player_1$ is represented by circles, $\Player_2$ by squares, $\Player_3$ by diamonds, and $\Player_4$ by hexagons. For each $i$, $t_i$ denotes the target of player $i$.
Edges that receive no probability weight in our example run of the algorithm are shown as dotted.}
\label{fig:reachability-example}
\end{figure}

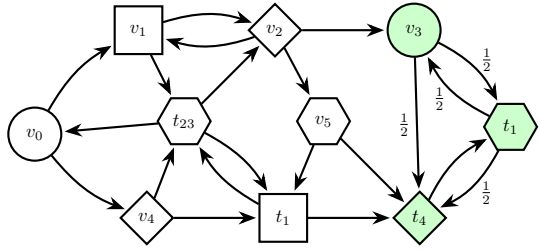
\begin{figure}[t]
\centering
\begin{tikzpicture}[
    scale=0.7,
    transform shape,
    >=Stealth,
    shorten >=1pt,
    auto,
    node distance=1.6cm,
    every state/.style={draw, thick, minimum size=8mm, inner sep=1pt}
]

\tikzset{
    p1/.style={draw, circle, thick, minimum size=10mm, inner sep=1pt},
    p2/.style={draw, rectangle, thick, minimum size=9mm, inner sep=1pt},
    p3/.style={draw, diamond, thick, minimum size=10mm, inner sep=1pt},
    p4/.style={draw, regular polygon, regular polygon sides=6, thick, minimum size=10mm, inner sep=1pt}
}

\node[p1] (0) {$v_0$};
\node[p2, above right=of 0] (1) {$v_1$};
\node[p3, right=of 1] (2) {$v_2$};
\node[p1, fill = green!20, right=of 2] (3) {$v_3$};

\node[p4, fill = green!20, below right=of 3] (9) {$t_{1}$};
\node[p3, fill = green!20,  below left=of 9] (8) {$t_4$};

\node[p2, left=of 8] (7) {$t_1$};
\node[p3, left=of 7] (6) {$v_4$};

\node[p4, below left=of 2] (4) {$t_{23}$};
\node[p4, right=of 4] (5) {$v_5$};

 {};

\path[->, thick]
    (0) edge[bend left=15]  (1)
    (0) edge[bend right=15] (6)

    (1) edge (4)
    (1) edge[bend left=18]  (2)
    
    (2) edge[bend left=18]  (1)
    (2) edge (3)
    (2) edge (5)

    (3) edge node[midway,left]{$\frac{1}{2}$} (8)
    (3) edge [bend left=18] node[midway ,above=4pt,right]{$\frac{1}{2}$} (9)

    (4) edge  (0)
    (4) edge (2)
    (4) edge [bend left=15] (7)

    (5) edge (8)
    (5) edge (7)
    
    (6) edge (7)
    (6) edge (4)

    (7) edge[bend left=15](4)
    (7) edge (8)

    (8) edge[bend left=18] (9)

    (9) edge [bend left=18] node[midway ,below=4pt,left]{$\frac{1}{2}$} (3)
    (9) edge[bend left=18] node[midway ,below=4pt,right]{$\frac{1}{2}$} (8);

\end{tikzpicture}
\caption{A four-player reachability game, where $t_i,$ denotes a target of $\Player_i$ and $t_{ij}$ denotes a target shared by $\Player_i$ and $\Player_j$.}
\label{fig:reachability-random}
\end{figure}

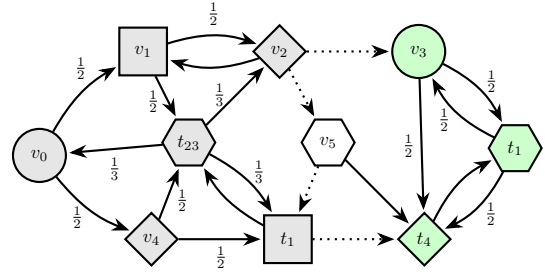
\begin{figure}[t]
\centering
\begin{tikzpicture}[
    scale=0.7,
    transform shape,
    >=Stealth,
    shorten >=1pt,
    auto,
    node distance=1.6cm,
    every state/.style={draw, thick, minimum size=8mm, inner sep=1pt}
]

\tikzset{
    p1/.style={draw, circle, thick, minimum size=10mm, inner sep=1pt},
    p2/.style={draw, rectangle, thick, minimum size=9mm, inner sep=1pt},
    p3/.style={draw, diamond, thick, minimum size=10mm, inner sep=1pt},
    p4/.style={draw, regular polygon, regular polygon sides=6, thick, minimum size=10mm, inner sep=1pt}
}

\node[p1, fill = gray!20] (0) {$v_0$};
\node[p2, fill = gray!20, above right=of 0] (1) {$v_1$};
\node[p3, fill = gray!20, right=of 1] (2) {$v_2$};
\node[p1, fill = green!20, right=of 2] (3) {$v_3$};

\node[p4, fill = green!20, below right=of 3] (9) {$t_{1}$};
\node[p3, fill = green!20,  below left=of 9] (8) {$t_4$};

\node[p2, fill = gray!20, left=of 8] (7) {$t_1$};
\node[p3, fill = gray!20, left=of 7] (6) {$v_4$};

\node[p4, fill = gray!20, below left=of 2] (4) {$t_{23}$};
\node[p4, right=of 4] (5) {$v_5$};

 {};

\path[->, thick]
    (0) edge[bend left=15] node[midway ,above=4pt] {$\frac{1}{2}$} (1)
    (0) edge[bend right=15] node[midway ,below=4pt, left] {$\frac{1}{2}$} (6)

    (1) edge node[midway ,above=4pt, below=2pt, left] {$\frac{1}{2}$} (4)
    (1) edge[bend left=18] node[midway ,above=2pt] {$\frac{1}{2}$} (2)
    
    (2) edge[bend left=18]  (1)
    (2) edge [dotted](3)
    (2) edge [dotted](5)

    (3) edge node[midway,left]{$\frac{1}{2}$} (8)
    (3) edge [bend left=18] node[midway ,above=4pt,right]{$\frac{1}{2}$} (9)

    (4) edge node[midway ,below=1pt] {$\frac{1}{3}$} (0)
    (4) edge node[midway ,left=2pt] {$\frac{1}{3}$}(2)
    (4) edge [bend left=15] node[midway ,above=4pt,right] {$\frac{1}{3}$}(7)

    (5) edge (8)
    (5) edge [dotted](7)
    
    (6) edge node[midway ,below] {$\frac{1}{2}$}(7)
    (6) edge node[midway ,below=4pt,right] {$\frac{1}{2}$} (4)

    (7) edge[bend left=15](4)
    (7) edge [dotted](8)

    (8) edge[bend left=18] (9)

    (9) edge [bend left=18] node[midway ,below=4pt,left]{$\frac{1}{2}$} (3)
    (9) edge[bend left=18] node[midway ,below=4pt,right]{$\frac{1}{2}$} (8);

\end{tikzpicture}
\caption{The reduced arena of (Figure~\ref{fig:reachability-random}) after removing edges (removed edges are shown as dotted lines) from evaluated vertices; a new leaf SCC emerges (highlighted in gray).}
\label{fig:reachability-reduced}
\end{figure}

\subsection{Safety Objectives}
While playing maximally mixed in leaf SCCs is optimal for players with reachability objectives, it is the worst opponent for safety players -- which curiously provides a different form of stability.

\begin{example}\label{ex:safety}
Let us consider the four-player game played on the arena shown in Figure~\ref{fig:safety-example}. Each player $\Player_i$ is associated with a set of safe vertices $S_i \subseteq V$, as depicted in the figure.

Similarly to \jamesColor{Example~\ref{ex:reach}}, we start from the leaf SCC. 
$\Player_4$ is not winning under uniform random play, but can win unilaterally from $s_{124}$ by fixing her strategy to move to $s_{24}$; we fix this edge and remove the edge to $s_3$.
This is safe to do because it guarantees that $\Player_4$ is winning regardless of how the other players play, as she wins surely against a maximally mixed strategy.
(This is the stability mentioned earlier: if the other players narrow down their options later, she will still win, so that these choices remain stable.)
We do \emph{not} evaluate at this point, we only fix the strategy of $\Player_4$.

This strategy for $\Player_4$ changes the structure of the game, with the new (only) leaf consisting of $s_{24}$ and $s_{124}$.
In this leaf component, there is no player who loses when all vertices in the leaf component are visited infinitely often, but can win unilaterally anywhere in this leaf component if all other players play this way.
We now fix this strategy for the complete leaf, knowing that it is stable, and evaluate it by marking $s_{24}$ and $s_{124}$ as winning for $\Player_2$ and $\Player_4$. 

We then select winning successors according to Rule~\ref{case:reach:1}, discard losing successors (where alternatives remain) according to Case~\ref{case:reach:2}, and evaluate vertices where the only successor is evaluated according to (Rule~\ref{case:reach:3}), much like for the reachability case.

In this example, we could evaluate $s_{3}$ next, as it has only one successor.
Vertex $s_{3}$ is only safe for $\Player_3$, and its only successor is marked as winning for $\Player_2$ and $\Player_4$, so that no player wins from $s_{3}$.
The remaining vertices can be evaluated in the following order: $v_3$, $s_{12}$, $v_2$,$s_{134}$, $v_3$, $s_{14}$, and $v_0$, as in Example~\ref{ex:reach}.
We depict a memoryless strategy profile $\sigma$ obtained in this way by solid arrows, while dotted arrows denote unused edges. 
\jamesColor{This profile is an ENE and an SPE from all vertices.}
\end{example}

Like in the reachability case, new leaves can occur during evaluation, and we turn to their evaluation when none of the Rules~\ref{case:reach:1}-\ref{case:reach:3} apply.

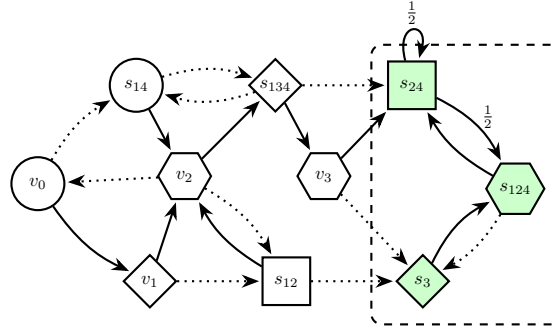
\begin{figure}[t]
\centering
\begin{tikzpicture}[
    scale=0.7,
    transform shape,
    >=Stealth,
    shorten >=1pt,
    auto,
    node distance=1.6cm,
    every state/.style={draw, thick, minimum size=8mm, inner sep=1pt}
]

\tikzset{
    p1/.style={draw, circle, thick, minimum size=10mm, inner sep=1pt},
    p2/.style={draw, rectangle, thick, minimum size=9mm, inner sep=1pt},
    p3/.style={draw, diamond, thick, minimum size=10mm, inner sep=1pt},
    p4/.style={draw, regular polygon, regular polygon sides=6, thick, minimum size=10mm, inner sep=1pt}
}

\node[p1] (0) {$v_0$};
\node[p1, above right=of 0] (1) {$s_{14}$};
\node[p3, right=of 1] (2) {$s_{134}$};
\node[p2, fill = green!20, right=of 2] (3) {$s_{24}$};

\node[p4, fill = green!20, below right=of 3] (9) {$s_{124}$};
\node[p3, fill = green!20,  below left=of 9] (8) {$s_{3}$};

\node[p2,left=of 8] (7) {$s_{12}$};
\node[p3, left=of 7] (6) {$v_1$};

\node[p4, below left=of 2] (4) {$v_2$};
\node[p4, right=of 4] (5) {$v_3$};

\node[draw, dashed, thick, rounded corners, fit=(3)(8)(9), inner sep=6pt] {};

\path[->, thick]
    (0) edge[bend left=15] [dotted](1)
    (0) edge[bend right=15] (6)

    (1) edge (4)
    (1) edge[bend left=18] [dotted](2)
    
    (2) edge[bend left=18] [dotted] (1)
    (2) edge [dotted] (3)
    (2) edge (5)

    (3) edge[loop above] node[midway,above]{$\frac{1}{2}$} (3)
    (3) edge [bend left=18] node[midway ,above=4pt,right]{$\frac{1}{2}$} (9)

    (4) edge [dotted](0)
    (4) edge (2)
    (4) edge [bend left=15] [dotted](7)

    (5) edge [dotted](8)
    (5) edge (3)

    (6) edge [dotted](7)
    (6) edge (4)

    (7) edge[bend left=15] (4)
    (7) edge [dotted](8)

    (8) edge[bend left=18] (9)

    (9) edge[bend left=15] (3)
    (9) edge[bend left=15] [dotted](8);

\end{tikzpicture}
\caption{A four-player safety game. Vertices $s_i$, $s_{ij}$, and $s_{ijk}$ are safe for $\Player_i$,$\Player_j$ and $\Player_k$ respectively, and $v_i$ is safe for all players. Edges that receive no probability weight in our example run of the algorithm are shown as dotted.}
\label{fig:safety-example}
\end{figure}

\subsection{0-2 Muller Objectives}
While we have seen that 1-3 Muller objectives may not have any memoryless NE, 0-2 Muller objectives share enough properties with reachability that we can treat them similarly.

Maximally mixing strategies in leaves are at the same time optimal for upwards closed winning sets (much like for reachability objectives) and the worst opponent for downwards closed winning sets (much like for safety objectives). As a result, we can treat them like safety games:
if a player almost-surely loses in a leaf when all players play uniformly at random, but can almost surely win unilaterally on a component within the leaf, they can do so using a positional strategy and we can fix it.
This is because this leaf cannot contain an upwards closed winning region of that player -- for if it would, then the leaf itself would be a winning region and all players playing uniformly at random would be almost surely winning for her.
Consequently, she can only win if she can force the game to eventually stay in some downward closed winning set.
This is a safety objective; in particular, if other players later remove some of their edges, the region visited can only shrink, and the downward closedness ensures that the resulting infinity set is winning for her.

\subsection{Constructing Nash Equilibria}
\label{sec:algo}
We now generalise the ideas presented earlier in this section to prove the existence of memoryless SPE in games with players that may have reachability objectives, safety objectives, and 0-2 Muller objectives.
Formally, the main result of this section is the following.
\begin{theorem}\label{theo:SPE}
Let $\G = (A, (\Omega_i)_{i\in\llb n\rrb})$ be a game where, for all $i\in\llb n\rrb$, $\Omega_i$ is either a reachability, safety or a 0-2 Muller objective (i.e., defined by a family $\mathcal{F}$ such that, for all $F\in\mathcal{F}$, either all subsets of $F$ are in $\mathcal{F}$ or all supersets of $F$ are in $\mathcal{F}$).
Then there exists a strategy profile $\sigma$ that is a memoryless randomised SPE from any vertex such that for all $i\in\llb n\rrb$ and all $v\in V$, $\mathbb{P}_{v}^\sigma(\Omega_i)\in\{0, 1\}$.
\end{theorem}

Let $\G = (A, (\Omega_i)_{i \in \llb n \rrb})$ be such a game. We provide an algorithm to construct this SPE.
In the following, we view a memoryless strategy profile as a function $\sigma\colon V\to \mathcal{D}(V)$.

We construct the strategy profile incrementally as a sequence $(\sigma_k)_{k \geq 0}$ of partially-defined strategy profiles, where each $\sigma_{k+1}$ extends $\sigma_k$.
We label the vertices $v$ in the domain of $\sigma_k$ by the set of players $W_v$ whose objectives are almost-surely satisfied from $v$ when playing according to $\sigma_k$ (this set is independent of $k$).
In some steps of the algorithm, we may not extend $\sigma_k$ and instead remove edges that do not lead to a winning scenario for the player who controls the vertex.

Initially, we let $\sigma_0$ be defined over the empty set.
We say that a vertex $v\in V$ is evaluated (at step $k$) if $\sigma_k(v)$ is defined, and say that $v$ is unevaluated otherwise.
The algorithm iteratively applies the following rules in arbitrary order until the domain of $\sigma_k$ is $V$.
\begin{enumerate}
    \item If there are no unevaluated vertices with evaluated successors, \label{rule:1}
    pick a leaf SCC $C$. We distinguish the following rules:
    \begin{enumerate}
        \item There exists a player $\Player_i$ whose objective is almost-surely falsified from $C$ when all players play uniformly at random, but who can win with positive probability from some vertex in $C$ when the other players play uniformly at random.
        In this case, the objective of $\Player_i$ is either a safety objective or a Muller objective given by $\mathcal{F}_i$ whose satisfaction is witnessed by a set $F\in\mathcal{F}_i$ whose subsets are all in $\mathcal{F}_i$.
        In either case, there exists a subset $C'\subseteq C$ and a positional strategy $\tau_i$ such that, for all strategy profiles $\tau_{-i}$ of the other players that do not exit $C$, all plays played according to $(\tau_i, \tau_{-i})$ starting in $C'$ remain in $C'$ and satisfy $\Omega_i$ (i.e., $\tau_i$ is a winning strategy in a zero-sum game and $C'$ is the set of vertices from which $\Player_i$ has a sure winning strategy).
        We fix such a pure memoryless winning strategy $\tau_i$ for $\Player_i$ over $C'$ and remove outgoing edges of $C'$ inconsistent with $\tau_i$.
        We let $\sigma_{k+1}=\sigma_k$.
        
        \item Otherwise, we extend $\sigma_{k}$ to $\sigma_{k+1}$ by choosing successors uniformly at random from all vertices in $C$.
        For all $v\in C$, we define $W_v$ to be the set of players whose objectives are satisfied when playing according to $\sigma_{k+1}$ (we explain how to determine this set below).
    \end{enumerate}

    \item If there exists $i\in\llb n\rrb$ and an unevaluated vertex $v \in V_i$ such that one of its successors $v'$ is evaluated and satisfies $i\in W_{v'}$, we define $\sigma_{k+1}(v) = v'$.\label{rule:2}
    We define $W_v$ as above and remove all other outgoing edges of $v$.

    \item If there exists a vertex $v$ that is unevaluated with a unique evaluated successor $v'$, we let $\sigma_{k+1}(v) = v'$. We define $W_v$ as above.\label{rule:3}

    \item Finally, if there exists $i\in\llb n\rrb$ and an unevaluated vertex $v \in V_i$ that has at least two successors, one of these successors $v'$ is evaluated and satisfies $i\notin W_{v'}$, we remove the edge from $v$ to $v'$ and let $\sigma_k = \sigma_{k+1}$.\label{rule:4}
\end{enumerate}

We let $\sigma$ denote an output of the above algorithm (it need not be unique).
The sets $W_v$ can be computed in practice as follows.
When playing uniformly at random in a leaf SCC $C$ due to Rule~\ref{rule:1}, the set $W_v$ for any $v\in C$ contains the players with a reachability objective with a target that intersects $C$, those with a safety objective for whom all vertices of $C$ are safe and those with a Muller objective for which $C$ is in the family defining the objective.
As all vertices of a leaf SCC $C$ are visited infinitely often almost-surely when playing at random, we remark that for all $v\in C$ and $i\notin W_v$, $\mathbb{P}^{\sigma}_v(\Omega_i) = 0$.
For a vertex $v$ such that $\sigma(v)$ is given by Rule~\ref{rule:2} or~\ref{rule:3}, for all $i\in\llb n\rrb$, whether $i\in W_v$ depends on the objective of $\Player_i$: if $\Player_i$ has a reachability objective, then $i\in W_v$ if, and only if, $v\in W_{\sigma(v)}$ or $v$ is in the target of $\Player_i$; if $\Player_i$ has a safety objective, then $i\in W_v$ if, and only if, $v\in W_{\sigma(v)}$ and $v$ is safe for $\Player_i$; if $\Player_i$ has a Muller objective, then $W_v = W_{\sigma(v)}$.
We obtain the following property of $\sigma$ by the above.

\begin{lemma}\label{lem:objective}
For all $i\in\llb n\rrb$ and all $v \in V$, $\mathbb{P}^{\sigma}_v(\Omega_i) \in \{0,1\}$.
\end{lemma}

We now prove that $\sigma$ is an ENE.
Our proof is by induction: we show that, for all $k\in \mathbb{N}$, any strategy $\sigma'$ that extends $\sigma_k$ is an NE from the vertices in the domain of $\sigma_k$.

\begin{lemma}\label{lem:dev}
Any strategy profile $\sigma$ constructed by the above algorithm is an ENE.
\end{lemma}

\begin{proof}
We prove the result by induction.
In the base case, $\sigma_0$ is defined over $\emptyset$ and therefore the claim is vacuously true.

We now assume by induction that any memoryless strategy profile extending $\sigma_k$ is an NE from every vertex the domain of $\sigma_k$.
Let $\sigma_{k+1}$ be the extension of $\sigma_k$ obtained after one additional iteration, and let $\sigma'$ be a (fully defined) memoryless profile extending $\sigma_{k+1}$.
We show that it is an NE from every vertex in the domain of $\sigma_{k+1}$.
It suffices to show that no player has a profitable deviation with respect to $\sigma'$ from vertices over which $\sigma_{k+1}$ is defined and $\sigma_{k}$ is not.
We consider three rules.

First, assume that $\sigma_{k+1}$ is obtained from $\sigma_k$ by playing uniformly at random over some set of vertices $C$ by Rule~\ref{rule:1} ($C$ is a leaf SCC that emerged during the algorithm).
We show that $\sigma'$ is an NE from all $v\in C$.
We remark that $W_v = W_{v'}$ for all $v, v'\in C$.
Let $i\in\llb n\rrb$.
If, for all $v\in C$, $i\in W_v$ (i.e., $\mathbb{P}^{\sigma'}_v(\Omega_i) = 1$),  then $\Player_i$ has no profitable deviation from any $v\in C$ with respect to $\sigma'$.
Assume therefore that $\mathbb{P}^{\sigma_{k+1}}_v(\Omega_i) = 0$ (cf. Lemma~\ref{lem:objective}).
By Rule~\ref{rule:1}(a), $\Player_i$ has no profitable deviation that remains in $C$.
If all outgoing edges of $V_i \cap C$ are in $C$, this shows that $\Player_i$ has no profitable deviation.
We now show that $\Player_i$ cannot improve their gain by using an edge from $V_i\cap C$ to $V\setminus C$.
Assume towards contradiction that this is the case, and let $(v, v')$ be such an edge.
This edge has been removed by Rule~\ref{rule:4}, and thus $i\notin W_{v'}$.
It follows that $\Player_i$ has a profitable deviation from $v'$ with respect to $\sigma'$.
This contradicts the induction assumption that $\sigma'$ is an NE from $v'$.

Second, we assume that $\sigma_{k+1}$ is obtained from $\sigma_k$ by Rule~\ref{rule:2} and let $v$ be the vertex that is added to the domain of $\sigma_{k}$.
Let $i\in\llb n\rrb$.
First, assume that $v\in V_i$.
If $\Player_i$ has a reachability or Muller objective, we have $i\in W_v$ and thus $\Player_i$ has no profitable deviation.
If $\Player_i$ has a safety objective, then $i\in W_v$ if and only if $v$ is safe for $\Player_i$; $\Player_i$ has no profitable deviation from $v$ regardless of the safety of $v$.
Assume now that $\Player_i$ does not control $v$.
If $\Player_i$ has a profitable deviation from $v$, then $\Player_i$ must have a profitable deviation from $\sigma_{k+1}(v)$.
By induction, there are no such profitable deviations.

Finally, we assume that $\sigma_{k+1}$ is obtained from $\sigma_k$ by Rule~\ref{rule:3}. 
Let $v$ be the vertex that is added to the domain of $\sigma_{k}$.
Let $i\in\llb n\rrb$.
We first consider the case $v\in V_i$.
We observe that no outgoing edges of $v$ could have been removed by Rule~\ref{rule:1}(a); otherwise, we would have $i\in W_v$.
It follows that all successors $v'$ of $v$ satisfy $i\notin W_{v'}$ by Rule~\ref{rule:4}.
If $\Player_i$ has a profitable deviation from $v$, it must be the case that $\Player_i$ has a profitable deviation from one of the successors, and there are no such deviations by induction.
We conclude similarly as in the previous case for players such that $v\notin V_i$.

We have shown that $\sigma'$ is an NE from all vertices in the domain of $\sigma_{k+1}$.
This shows that $\sigma$ is an ENE.
\end{proof}

We now prove Theorem~\ref{theo:SPE} by using Lemma~\ref{lem:dev}.

\begin{proof}[Proof of Theorem~\ref{theo:SPE}]
Let $\sigma=(\sigma_1,\sigma_2,\ldots,\sigma_n)$ be a strategy profile constructed by the algorithm. 
By Lemma~\ref{lem:objective}, we have $\mathbb{P}^{\sigma}_v(\Omega_i) \in \{0,1\}$ for all $i\in\llb n\rrb$ and all $v \in V$.

It remains to prove that $\sigma$ is an SPE from all vertices.
We prove it by contradiction.
Assume that there exists a history $hv$ (with $h$ a history and $v\in V$), a player $\Player_i$ and a profitable deviation $\sigma'_i$ for $\Player_i$ with respect to $\sigma$ from $v$ in the subgame $\G_{h}$.
Then $\sigma'_i$ is also a profitable deviation from $v$ with respect to $\sigma$ in $\G$ (due to the objectives we consider), contradicting Lemma~\ref{lem:dev}.
This shows that $\sigma$ is an SPE from all initial vertices.
\end{proof}

We now highlight a restricted class of games with reachability and safety objectives in which we can show the existence of pure memoryless strategy profiles that are SPEs from all vertices by adapting our construction.
To identify this class, we reflect on when randomisation is used in our construction, and when it is absent or can be omitted.
Randomised decisions are only made in leaf SCCs that emerge during the algorithm.
Its purpose is to enforce stability by visiting all targets within the SCC for the reachability players and by violating the objectives of the safety players for whom there is an unsafe vertex in the SCC.

We identify two situations in which randomisation is not used or is not helpful in achieving stability when our algorithm builds a strategy profile.
First, if all vertices in a leaf SCC have a single successor, then all strategies of this SCC are pure by default. 
In particular, in absorbing vertices (i.e., vertices who are their own unique successor), randomisation is not used.
Second, if no targets and no unsafe states occur in a leaf SCC, fixing any pure strategy on this SCC leads to a behaviour that is as stable as a maximally mixed strategy.
On the one hand, for the safety players, plays that stay in this leaf SCC satisfy their safety objective, therefore they have no profitable deviation.
On the other hand, for the reachability players, they cannot win within the SCC and none of the deleted edges that exit the leaf SCC can lead to a profitable deviation by construction (as explained in the proof of Lemma~\ref{lem:dev}).

These observations suggest that if all targets and all unsafe states are absorbing, we can construct pure memoryless strategy profiles that are SPEs from all vertices by changing Rule~\ref{rule:1}(b) to fix an arbitrary outgoing edge in all vertices of leaf SCCs.
This yields the following result.
\begin{theorem}\label{theo:absorbing}
Let $\G = (A, (\Omega_i)_{i\in\llb n\rrb})$ be a game where, for all $i\in\llb n\rrb$, $\Omega_i$ is either a reachability or a safety objective such that all unsafe vertices and target vertices are absorbing.
Then there exists a strategy profile $\sigma$ that is a memoryless pure SPE from any vertex.
\end{theorem}
Theorem~\ref{theo:absorbing} slightly generalises the main result of~\cite{Kehagias22}, which states this existence result for NE.
The proof approach used in~\cite{Kehagias22} is not algorithmic: the main idea is to consider a discounted version of the game and then transform a memoryless equilibrium into a positional one.
While the existence of memoryless equilibria in discounted games is known~\cite{Fink64}, the proof is not constructive and relies on Kakutani's fixed point theorem~\cite{kakutani41}.
In contrast, the adaptation of our algorithm outlined above provides a direct construction of an equilibrium.

\section{Conclusion}
In this paper, we have constructed memoryless randomised ENEs in 
turn-based games where all players have reachability, safety, or simple 
Muller objectives whose winning sets are up- or downward closed (the 0-2 class in the Mostowski hierarchy). We have first shown that randomisation is necessary even for simple combinations of reachability and safety objectives, and that memoryless NEs may fail to exist for 1-3 Muller objectives.
Our algorithm constructs an ENE in polynomial time. We proved that every ENE is an SPE in this setting, establishing the existence of memoryless randomised SPEs for this class of games. It remains open whether games in which all players have only reachability objectives, or only safety objectives, admit positional NEs, ENEs, or SPEs.

\bibliography{references.bib}

\end{document}